\documentclass[letterpaper, 10 pt, conference]{ieeeconf}  

\IEEEoverridecommandlockouts                              

\usepackage{mathrsfs}
\usepackage{xcolor}

\usepackage{amsmath,amssymb,amsfonts}
\usepackage{graphicx}
\usepackage{algorithm,algorithmic}

\usepackage{textcomp}

\usepackage{setspace}
\usepackage{lipsum}

\usepackage{mathtools}
\usepackage{enumerate}

\usepackage{caption}
\usepackage{subcaption}
\usepackage{algorithm}
\usepackage{float}
\newtheorem{definition}{Definition}
\newtheorem{theorem}{Theorem}

\newtheorem{lemma}{Lemma}

\newtheorem{proposition}{Proposition}
\usepackage{hyperref}

\makeatletter
\let\NAT@parse\undefined
\makeatother
\usepackage{cite}

\newtheorem{innerexample}{Example}

\newcommand{\abs}[1]{\left\lvert#1\right\rvert}

\newcommand{\R}{\mathbb{R}}

    {\begin{innerexample}[#1]\pushQED{\qed}}%
    {\popQED\end{innerexample}}
\title{\LARGE \bf
Modular Design of Strict Control Lyapunov Functions for Global Stabilization of the Unicycle in Polar Coordinates
}

\author{Velimir Todorovski, Kwang Hak Kim, and Miroslav Krstić
\thanks{This work was supported in part by the Office of Naval Research under Grant No. N00014-23-1-2376, in part by the Air Force Office of Scientific Research under Grant No. FA9550-23-1-0535, and in
part by the National Science Foundation under Grant No. ECCS-2151525. The results and opinions in this paper are solely of the authors and do not reflect the position or the policy of the U.S. Government or the National Science
Foundation.}
\thanks{V. Todorovski, K. Kim,  and M. Krstić are with the Department of Mechanical and Aerospace Engineering, UC San Diego, 9500 Gilman Drive, La Jolla, CA, 92093-0411, {\tt\small \{vtodorovski,kwk001,krstic\}@ucsd.edu}}%
}

\begin{document}

\maketitle
\thispagestyle{empty}
\pagestyle{empty}

\begin{abstract}
Since the mid-1990s, it has been known that, unlike in Cartesian form where Brockett’s condition rules out static feedback stabilization, the unicycle is globally asymptotically stabilizable by smooth feedback in polar coordinates. In this note, we introduce a modular framework for designing smooth feedback laws that achieve global asymptotic stabilization in polar coordinates. These laws are bidirectional, enabling efficient parking maneuvers, and are paired with families of strict control Lyapunov functions (CLFs) constructed in a modular fashion. The resulting CLFs guarantee global asymptotic stability with explicit convergence rates and include barrier variants that yield “almost global’’ stabilization, excluding only zero-measure subsets of the rotation manifolds. The strictness of the CLFs is further leveraged in our companion paper, where we develop inverse-optimal redesigns with meaningful cost functions and infinite gain margins.
\end{abstract}


\section{Introduction}

The unicycle, though globally controllable, cannot be asymptotically stabilized by any continuous or discontinuous time-invariant feedback, as shown by Brockett \cite{brockett1983asymptotic} and extended by Ryan, as well as Coron and Rosier \cite{ryan1994brockett,coron1994relation}. The Brockett–Ryan-Coron–Rosier conditions highlight fundamental topological obstructions to stabilizing nonholonomic systems.

This challenge has motivated diverse strategies to bypass the Brockett–Ryan-Coron–Rosier obstruction. One line of work introduces time-varying feedback, achieving global asymptotic stabilization of the unicycle \cite{samson1990velocity,coron1992global_controllable,pomet1992explicit,d2019small}, but such controllers are sensitive to delays, prone to synchronization issues, and often produce transient oscillations. 
%
Another line leverages hysteresis-based hybrid feedback with robustness refinements \cite{hespanha1999_hybrid_stabilization,prieur2003robust,prieur2005robust}, though these often cause zig-zagging motions, demand high actuation effort, and are difficult to implement in practice.

An alternative approach relies on coordinate transformations that introduce singularities in the transformed coordinates. 
A notable example is Astolfi’s $\sigma$-processes \cite{astolfi1995exponential,astolfi1996discontinuous}, which achieve continuous exponential stabilization in the transformed coordinates, though the transformation excludes certain initial conditions—such as states on the $x$-axis or aligned with the target.
Perhaps, a more physically meaningful transformation is to polar coordinates \cite{astolfi1999exponential}, which naturally encodes distance and relative heading to the target. Using this representation, Aicardi et al.\cite{aicardi1995} design a passivity-based feedback law with bidirectional forward velocity that achieves global asymptotic stabilization of the unicycle. Stability is established using a non-strict Lyapunov function and Barbalat’s lemma, precluding constructive $\mathcal{KL}$- estimates. Building on this work, Han and Wang\cite{wang24_force_controlled_safestable} introduce a strict Lyapunov function valid on large compact sets, while Restrepo et al.~\cite{restrepo2020leader} achieve global exponential stability through a backstepping design with a strict CLF. However, their approach sacrifices modularity—complicating extensions to barrier CLFs—and restricts the unicycle to unidirectional motion, potentially leading to less efficient parking trajectories.

In this work, we develop multitudes of feedback laws in polar coordinates for the parking problem of unicycle vehicles with bidirectional forward velocity. The design follows a modular structure: the forward velocity is selected to decouple the distance state from the angular dynamics, which are subsequently regulated by steering through passivity and backstepping techniques. The control laws achieve global asymptotic stabilization, supported by strict CLFs built in a composite fashion that expand the variety of attainable Lyapunov function designs and provide constructive $\mathcal{KL}$-estimates of convergence rates. Within this modular framework, barrier CLFs are developed to confine states to meaningful intervals, preventing angular wind-up and enforcing safety constraints. 
 We note that due to space restrictions, some proofs have been omitted and will be included in the extended journal version of this paper.

\raggedbottom
\section{Modular Feedback Design for the Unicycle}
We consider the unicycle model 
$\dot{x} =  v \cos\theta, \quad 
\dot{y} = v \sin\theta,
\quad \dot{\theta} = \omega \,,$
where $(x(t),y(t)) \in \R^2$ is the position of the unicycle in Cartesian coordinates, $\theta(t) \in \mathbb{R}$ is the heading angle, $v$ is the forward, and $\omega$ is the angular velocity input. The unicycle can be represented in polar coordinates with the transformation given in Fig.~\ref{fig:unicycle_cord} as
\begin{subequations}\label{eq:unicycle_polar_closed_loop-Gv-1}
    \begin{align}
    \dot{\rho} &= 
    -v \cos\gamma\label{eq:unicycle_polar_rhodot}\\
    \dot{\delta} &= v  \frac{\sin\gamma}{\rho}\label{eq:unicycle_polar_deltadot}\\
    \dot{\gamma} &= v \frac{\sin\gamma}{\rho} -\omega \label{eq:unicycle_polar_gammadot} \,.
    \end{align} 
\end{subequations}
The polar coordinates are defined as the distance \(\rho=\sqrt{x^2+y^2}\), the polar angle \(\delta=\text{{\rm atan2}}(y,x)+\pi\), and the line-of-sight (LoS) angle \(\gamma=\delta-\theta\). 
The unicycle model 
in Cartesian coordinates admits no continuous stabilizer due to the Brockett-Ryan–Coron–Rosier conditions. In contrast, its polar representation \eqref{eq:unicycle_polar_closed_loop-Gv-1} avoids this obstruction (due to singularity $\rho=0$), enabling smooth static feedback laws for $\rho>0$. However, when expressed back in Cartesian coordinates, discontinuities persist, notably from $\text{atan2}(y,x)$ along ${x<0,y=0}$.

\label{sec-preliminaries}
\paragraph{Forward velocity feedback}
\label{sec:forward_velocity_feedback}
For stabilization alone of the system \eqref{eq:unicycle_polar_closed_loop-Gv-1}, it is not evident that one can choose a feedback law that is either significantly better performing or simpler than the feedback
\begin{equation}
\label{eq-basic-v-control}
\fbox{$v =  k_1 \rho \cos\gamma$} = -k_1(x\cos\theta+y\sin\theta)\,,  \quad k_1>0\,.
\end{equation}
The feedback \eqref{eq-basic-v-control} yields the system consisting of the closed-loop subsystem 
\begin{align}
\label{eq:unicycle_polar_closed_loop-Gv-2}
\dot{\rho} &=  
- k_1 \rho \cos^2(\gamma)\,,
\end{align} 
for which $\rho(t)$ exponentially converges to zero when $\gamma\neq 0$, 
and, on the set $\rho>0$, at which the cancellation $\rho/\rho=1$ is valid, the subsystem 
\begin{subequations}
\label{eq:unicycle_polar_closed_loop-Gv-3}
\begin{align}
\dot{\delta} &=  \frac{k_1}{2} \sin(2\gamma) \label{eq:unicycle_polar_closed_loop-Gv-3n} \\
\dot{\gamma} &= \frac{k_1}{2} \sin(2\gamma) -\omega \label{eq:unicycle_polar_closed_loop-Gv-3a} \,,
\end{align} 
\end{subequations}
for which a feedback law depending only on $(\delta,\gamma)$, and applied by the steering input $\omega$, needs to be designed. 
By inspection of \eqref{eq:unicycle_polar_closed_loop-Gv-3a}, it is clear that, when the vehicle is not steered, namely, when $\omega=0$, the $\gamma$-subsystem is unstable at the equilibrium $\gamma=0$ and exponentially stable at the equilibrium $\gamma = \pm\pi/2$ (modulo $2\pi$). In other words, the vehicle, under the basic forward velocity feedback \eqref{eq-basic-v-control}, moves along a straight line and settles at a position at which the target position is exactly to its left or to its right, at a distance that depends on the initial conditions $(\rho_0,\gamma_0)$, and does so irrespective of the target heading. The consequence of these observations is that the term $\frac{k_1}{2} \sin(2\gamma)$ in \eqref{eq:unicycle_polar_closed_loop-Gv-3a} is destabilizing and it cannot be tolerated. We cancel this term and introduce a control law
\begin{equation}
\label{eq-omega-general}
\fbox{$\omega = \dfrac{k_1}{2} \sin(2\gamma) +\tilde\omega$}
\end{equation}
where the control $\tilde\omega$ has to be designed for the system
\begin{subequations}
\label{eq:unicycle_polar_closed_loop-Gv-3-pass}
\begin{eqnarray}
\label{eq:unicycle_polar_closed_loop-Gv-3n-pass}
\dot{\delta} &=&  \frac{k_1}{2} \sin(2\gamma)\\ \label{eq:unicycle_polar_closed_loop-Gv-3a-pass}
\dot{\gamma} &=& -\tilde\omega\,.
\end{eqnarray}
\end{subequations}


\begin{figure}[t]
\centering
\includegraphics[width=0.5\linewidth]{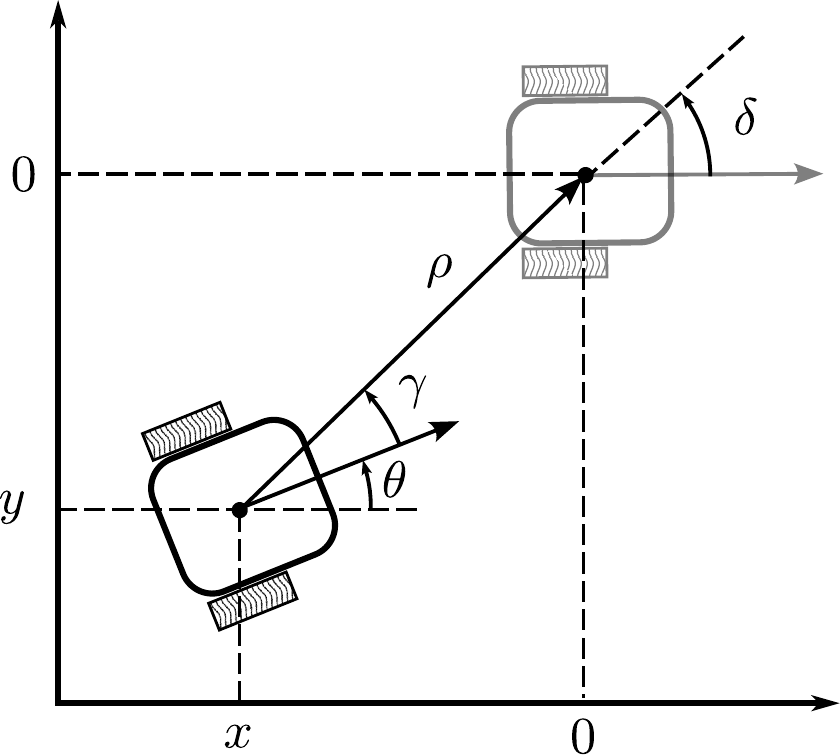}
\caption{Unicycle orientation \((x,y,\theta)\) with respect to the goal state \((0,0,0)\), along with the polar coordinate transformation \((x,y,\theta)\mapsto (\rho,\delta,\gamma)\). 
}
\label{fig:unicycle_cord}
\end{figure}
\paragraph{State spaces and stability definitions}
Before presenting the steering laws $\tilde{\omega}$, we introduce the state spaces 
\begin{equation}
\mathcal{S} := 
\left\{ \rho> 0\right\}\times\mathcal{T}\,,\quad \mathcal{T}:= \left\{ \delta \in\mathbb{R}, \gamma \in\mathbb{R}  \right\} 
\label{eq:ss_S}
\end{equation}
\begin{equation}
    \mathcal{S}_1 := \left\{ \rho> 0\right\}\times\mathcal{T}_1\,,\quad \mathcal{T}_1:= \left\{  \abs{\delta} <  \pi, \gamma \in\mathbb{R} \right\} \  \label{eq:ss_S1}
\end{equation}
\begin{equation}
   \mathcal{S}_2 := \left\{ \rho> 0\right\}\times\mathcal{T}_2\,,\quad \mathcal{T}_2:= \left\{ \delta \in\mathbb{R}, \abs{\gamma} < \pi \right\}\  \label{eq:ss_S2} 
\end{equation}
\begin{equation}
    \mathcal{S}_3 := \left\{ \rho> 0\right\}\times\mathcal{T}_3\,,\quad \mathcal{T}_3:= \left\{ \abs{\delta} < \pi, \abs{\gamma} < \pi\right\} \label{eq:ss_S3}
\end{equation}
 with the associated metrics
\begin{align*}
|(\rho,\delta,\gamma)|_{\mathcal{S}}   &\coloneqq \rho+|(\delta,\gamma)|_{\mathcal{T}} 
      :=  \rho + |\delta| + |\gamma|  
      \\
|(\rho,\delta,\gamma)|_{\mathcal{S}_1} &\coloneqq \rho+|(\delta,\gamma)|_{\mathcal{T}_1} 
      := \rho+ 2\tan\!\frac{|\delta|}{2} + |\gamma|  
      \\
|(\rho,\delta,\gamma)|_{\mathcal{S}_2} &\coloneqq \rho+|(\delta,\gamma)|_{\mathcal{T}_2} 
      := \rho+ |\delta| + 2\tan\!\frac{|\gamma|}{2}  
      \\
|(\rho,\delta,\gamma)|_{\mathcal{S}_3} &\coloneqq \rho+|(\delta,\gamma)|_{\mathcal{T}_3} 
      := \rho+2\tan\!\frac{|\delta|}{2} + 2\tan\!\frac{|\gamma|}{2}. 
\end{align*}


\begin{definition}
\label{def-our-GAS}
Consider the system \eqref{eq:unicycle_polar_closed_loop-Gv-2}, \eqref{eq:unicycle_polar_closed_loop-Gv-3-pass} with a feedback law $\tilde\omega(\gamma,\delta)$ that is continuous on a state space $\mathcal{Q}$ with respect to its metric. 
If there exists a class $\mathcal{KL}$ function $\beta$ such that, for all $t\geq 0$, it holds that $|(\rho(t),\delta(t), \gamma(t))|_{\mathcal{Q}}\leq \beta\left(|(\rho_0,\delta_0, \gamma_0)|_{\mathcal{Q}},t\right)$, we say that the 
point $\rho=\delta = \gamma=0$ is {\em globally asymptotically stable on $\mathcal{Q}$} (GAS on $\mathcal{Q}$). 
\end{definition}

With this definition, expressed explicitly in terms of a class $\mathcal{KL}$ bound, we circumvent the problem that stability is conventionally defined for equilibria away from the boundary of the state space. 

\paragraph{Composite Lyapunov functions}
In order to show the stability properties defined above, we generate a multitude of CLFs in the state $(\rho,\delta,\gamma)$. These CLFs are built ``modularly,'' since the system \eqref{eq:unicycle_polar_closed_loop-Gv-3-pass} with $\tilde\omega(\delta,\gamma)$ is independent of $\rho$ and, hence, the CLF for the $(\delta,\gamma)$-subsystem \eqref{eq:unicycle_polar_closed_loop-Gv-3}, denoted $V_{\delta\gamma}(\delta,\gamma)$, depends only on  $(\delta,\gamma)$, whereas the CLF for \eqref{eq:unicycle_polar_closed_loop-Gv-2}
naturally depends only on $\rho$, with a nonzero $\gamma(t)$ regarded simply as a distrubance. 
The following proposition gives examples of modularly designed CLFs $V(\rho,\delta,\gamma) = \mathcal{V}(\rho^2,V_{\delta\gamma})$, where, due to the scalar nature of \eqref{eq:unicycle_polar_closed_loop-Gv-2}, its CLF is taken, w.l.o.g, as $\rho^2$. 

\begin{proposition}[{Composite Lyapunov functions}] 
\label{prop:composite_Lyap_function}
Consider 
any continuously differentiable function $(\delta,\gamma) \mapsto V_{\delta\gamma}$ where $(\delta,\gamma)$ belong to any state space $\hat{\mathcal{T}}\in\{\mathcal{T},\mathcal{T}_1,\mathcal{T}_2,\mathcal{T}_3\}$
and  such that 
$ \alpha_1(|(\delta, \gamma)|_{\hat{\mathcal{T}}}) \le V_{\delta\gamma}(\delta,\gamma) \le \alpha_2(|(\delta, \gamma)|_{\hat{\mathcal{T}}})$ where $\alpha_1, \alpha_2$ are class $\mathcal{K}_{\infty}$ functions. 
Let $(\delta,\gamma) \mapsto\tilde\omega$ be a continuous function such that 
\begin{equation}
\dot V_{\delta\gamma} := \dfrac{\partial V_{\delta\gamma}}{\partial\delta}
\frac{k_1\sin(2\gamma)}{2}  -\dfrac{\partial V_{\delta\gamma}}{\partial\gamma}\tilde\omega(\delta,\gamma) \leq -\alpha_{\delta\gamma}(V_{\delta\gamma}) \label{eq:V_dot_delta_gamma}
\end{equation}
for some class $\mathcal{K}$ function $\alpha_{\delta\gamma}$. 
Then for {\em any} function $(r,s) \mapsto \mathcal{V}$ such that 
\begin{enumerate}
\item $\mathcal{V}(0,0)=0$ and $\mathcal{V}(r, s)>0$ if $r > 0 $ or $s > 0$, \label{it:pos_def_composite_Lyap_function} 
\item $\lim_{r+s\rightarrow\infty} \mathcal{V}(r,s) = \infty$,
\item $\dfrac{\partial \mathcal{V} }{\partial r }(r,s)>0$ and $\dfrac{\partial \mathcal{V} }{\partial s}(r,s)>0$  if $r>0$ or $s > 0$, \label{it:positive_partials}
\end{enumerate}
including the particular functions
 $\mathcal{V}(r,s)= r+s$, $\mathcal{V}(r,s)= \ln(1+r) + s$, $\mathcal{V}(r,s)= (1+r){\rm e}^{s} -1$,  
for both of the `composite' Lyapunov functions $V(\rho,\delta,\gamma) = \mathcal{V}\left(\rho^2,V_{\delta\gamma}(\delta,\gamma)\right)$ and $V(\rho,\delta,\gamma) = \mathcal{V}\left(V_{\delta\gamma}(\delta,\gamma), \rho^2 \right)$ there exist the respective (distinct) 
triplets of functions $\bar\alpha_1,\bar\alpha_2\in\mathcal{K}_\infty$, $\alpha\in\mathcal{K}$
such that,  for all $(\rho,\delta,\gamma)$ in the state space $\hat{\mathcal{S}}\in\{\mathcal{S},\mathcal{S}_1,\mathcal{S}_2,\mathcal{S}_3\}$, it holds that  
$ \bar\alpha_1(|(\rho,\delta, \gamma)|_{\hat{\mathcal{S}}}) \le V(\rho,\delta,\gamma) \le \bar\alpha_2(|(\rho,\delta, \gamma)|_{\hat{\mathcal{S}}})$ and 
\begin{align}
\dot{V}(\rho,\delta,\gamma) := 
-\dfrac{\partial V }{\partial \rho}
k_1\rho^2&\cos^2\gamma  + 
\dfrac{\partial V}{\partial \delta}
\frac{k_1 \sin(2\gamma)}{2} 
 \nonumber\\
&-\dfrac{\partial V}{\partial\gamma}\tilde\omega(\delta,\gamma) \leq -\alpha(V)\,. \label{eq:V_dot_rho_delta_gamma}
\end{align}

\end{proposition}
The composite Lypunov functions in Proposition \ref{prop:composite_Lyap_function} are actually CLFs for \eqref{eq:unicycle_polar_closed_loop-Gv-1}, in accordance with the following definition. 
\begin{definition}[CLF for the unicycle \eqref{eq:unicycle_polar_closed_loop-Gv-1}]
\label{def-CLF}
A continuously differentiable function $(\rho,\delta,\gamma)\mapsto V$ 
is a \textit{control Lyapunov function} (CLF) for \eqref{eq:unicycle_polar_closed_loop-Gv-1} if it has the following properties.
\begin{enumerate}
\item There exist class $\mathcal{K}$ functions  $(\bar\alpha_1,\bar\alpha_2)$ such that,  for all $(\rho,\delta,\gamma)$ in 
$\Sigma = \{\rho\geq 0\}\times \hat{\mathcal{T}}$, where $\hat{\mathcal{T}}\in\{\mathcal{T},\mathcal{T}_1\}$, 
$ \bar\alpha_1(|(\rho,\delta, \gamma)|_{\hat{\mathcal{S}}}) \le V(\rho,\delta,\gamma) \le \bar\alpha_2(|(\rho,\delta, \gamma)|_{\hat{\mathcal{S}}})$, where $\hat{\mathcal{S}} = \{\rho> 0\}\times \hat{\mathcal{T}}$. \label{it:clf_prop_1}
\item There exists $(v/\rho,\omega)\in \mathbb{R}^2$ such that
\begin{equation*}
\left[-\dfrac{\partial V}{\partial\rho}\rho\cos\gamma + \left(\dfrac{\partial V}{\partial\delta}+\dfrac{\partial V}{\partial\gamma}\right)\sin\gamma \right]\dfrac{v}{\rho} - \dfrac{\partial V}{\partial\gamma}\omega < 0 \,,
\end{equation*}
for all $(\rho,\delta,\gamma)\neq (0,0,0)$ in $\Sigma$. \label{it:clf_prop_2}
\end{enumerate}
\end{definition}

\section{Families of CLFs for the Unicycle}

In this section we summarize the principal control designs for the $(\delta,\gamma)$-subsystem \eqref{eq:unicycle_polar_closed_loop-Gv-3-pass}: a pair of backstepping designs for the state spaces $\mathcal{T}$ and $\mathcal{T}_1$ and two passivity-based designs for the state spaces $\mathcal{T}_2$ and $\mathcal{T}_3$. Each steering controller admits a function $V_{\delta\gamma}$ which can be viewed as a global strict CLFs for \eqref{eq:unicycle_polar_closed_loop-Gv-3-pass} on its respective state space. Building on Proposition~\ref{prop:composite_Lyap_function} and the designed $V_{\delta\gamma}$-CLFs for subsystem \eqref{eq:unicycle_polar_closed_loop-Gv-3-pass}, we obtain multiple families of CLFs for \eqref{eq:unicycle_polar_closed_loop-Gv-1}.

\subsection{Backstepping-based design}
Let us regard the system \eqref{eq:unicycle_polar_closed_loop-Gv-3-pass} as a double integrator chain, with the sinusoidal nonlinearity acting between the two integrators and limiting the magnitude of the input that acts on $\dot\delta$. Such a nonlinear integrator chain provides an opportunity for an unconventional application of the backstepping method. The next two controllers are designed with such an approach. 
Let us start with the identity
\begin{equation}
\sin(\arctan(x)-y) = \frac{x\cos y - \sin y}{\sqrt{1+x^2}},
\end{equation}
Considering \eqref{eq:unicycle_polar_closed_loop-Gv-3n-pass}, 
we note that the would-be feedback for the LoS-angle (the `stabilizing function') 
\begin{equation}
\gamma(\delta) = -\frac{1}{2} \arctan(2k_2 \Delta(\delta)) 
\end{equation}
where the function $\Delta(\delta)$  is defined later for each of the state-space \eqref{eq:ss_S} and \eqref{eq:ss_S1}, respectively. With $k_2 > 0$, substituting 
$\gamma(\delta)$ in \eqref{eq:unicycle_polar_closed_loop-Gv-3n-pass} results in the globally asymptotically and locally exponentially stable $\delta$-dynamics
\begin{equation}
\dot{\delta} = k_1\frac{\sin(2\gamma(\delta))}{2} 
= - \frac{k_1 k_2 \Delta}{\sqrt{1+4 k_2 ^2\Delta^2}}
\end{equation} 
if $\Delta(\delta)$ has the same sign as $\delta$.
Now, we introduce a backstepping change of the $\gamma$-variable
\begin{align}
\label{eq:backstepping_z}
\fbox{$\displaystyle z = \gamma + \frac{1}{2}\arctan(2 k_2 \Delta)$}
\end{align}
The backstepping transformation as well as the controllers employ the following continuous, bounded function:
\begin{align}
&\psi(z,\gamma) = \frac{\sin(2z-2\gamma) +\sin(2\gamma)}{2z}
\nonumber\\
=& \frac{1}{\sqrt{1 + 4k_2^2 \Delta^2}} \left(\,\frac{\sin(2z)}{2z}
+ 2k_2 \Delta \frac{1 - \cos(2z)}{2z}\right)
 \label{eq:psi_definition}
\end{align}
that has the property that $\psi(0,\gamma)  = \cos(2\gamma)$ and, in particular, that $\psi(0,0) = \max\psi(z,\gamma)= \psi(0,n\pi) = 1$ for all integer $n$. If, instead of the ``unconventional, bounded integrator chain'' \eqref{eq:unicycle_polar_closed_loop-Gv-3-pass} one had $\dot\delta = k_1 \gamma$, the function $\psi$ would be simply $\psi(z,\gamma)\equiv 1$, namely, the virtual input coefficient.

\paragraph{GloBa ($\mathcal{T}$) and BAR-FLi ($\mathcal{T}_1$) controllers}
\begin{theorem}
\label{thm:backstepping_theorem}
Consider the system \eqref{eq:unicycle_polar_closed_loop-Gv-1} in closed-loop with \eqref{eq-basic-v-control}, \eqref{eq-omega-general}, and 
\begin{equation}
\tilde\omega = k_4 z + \dfrac{\partial \Delta}{\partial \delta} \left( \frac{k_1k_2 \sin(2\gamma)}{2(1+4k_2^2 \Delta^2)} + k_3 \psi(z,\gamma)\Delta \right),
\label{eq:backstepping_controllers}
\end{equation}
where $k_1,k_2,k_3,k_4 > 0$, and $z$ and $\psi(z,\gamma)$ are defined in \eqref{eq:backstepping_z} and \eqref{eq:psi_definition}, respectively. Let $\Delta$ be given by either
\begin{align}
\Delta &= \delta, \label{eq:Delta_globa} \\
\Delta &= 2 \tan \frac{\delta}{2}. \label{eq:Delta_Barfli}
\end{align}
Then, the origin $\rho = \delta = \gamma = 0$ is globally asymptotically stable on $\mathcal{S}$ in the case of \eqref{eq:Delta_globa}, and on $\mathcal{S}_1$ in the case of \eqref{eq:Delta_Barfli}, in the sense of Definition~\ref{def-our-GAS}.
Furthermore, all the composite Lyapunov functions $V(\rho, \delta, \gamma) = \mathcal{V}(\rho^2, V_{\delta \gamma})$ and $V(\rho, \delta, \gamma) = \mathcal{V}(V_{\delta \gamma},\rho^2)$, for all functions $\mathcal{V}$ satisfying the conditions in Proposition \ref{prop:composite_Lyap_function}, and with $V_{\delta \gamma}$ for both cases of $\Delta$ defined as
\begin{equation}
\label{eq:V_CLF_backstepping}
V_{\delta \gamma}(\delta, \gamma) = 
\Delta^2 + q^2z^2\,, \quad q = \sqrt{\frac{k_1}{k_3}}\,,
\end{equation}
are (globally) strict CLFs for \eqref{eq:unicycle_polar_closed_loop-Gv-1} with respect to the input pair $(v/\rho,\omega)$ in the sense of Definition \ref{def-CLF}. 
\end{theorem}

\begin{proof}
Taking into account \eqref{eq:unicycle_polar_closed_loop-Gv-3-pass} with \eqref{eq:backstepping_z}, \eqref{eq:psi_definition} and \eqref{eq:backstepping_controllers},  
we obtain
\begin{subequations}
\begin{align}
\dot{\delta} &= - k_1 \left( \frac{k_2 \Delta}{\sqrt{1+4k_2^2\Delta^2}}
- z \psi(z,\gamma) \right) \,, \label{eq:delta_globa} \\
\dot{z} &=  -k_4 z -k_3 \frac{\partial \Delta}{\partial \delta} \psi(z,\gamma) \Delta \,.  \label{eq:z_globa}
\end{align}
\label{eq:backstepping_closed_loop}
\end{subequations}
Since the Lyapunov function \eqref{eq:V_CLF_backstepping} with either \eqref{eq:Delta_globa} or \eqref{eq:Delta_Barfli}, is positive definite on $\mathcal{T}$ and $\mathcal{T}_1$ respectively, there exists class $\mathcal{K}_{\infty}$-functions $\alpha_1$, $\alpha_2$ such that $\alpha_1(\abs{(\delta,\gamma)}_{\hat{\mathcal{T}}}) \le V_{\delta \gamma }(\delta, \gamma) \le \alpha_2(\abs{(\delta,\gamma)}_{\hat{\mathcal{T}}})$ where $\hat{\mathcal{T}} = \{\mathcal{T}, \mathcal{T}_1\}$ is dependent on the choice of $\Delta$.
The time derivative \eqref{eq:V_CLF_backstepping} along the solutions of \eqref{eq:backstepping_closed_loop} is  
\begin{equation}
    \dot{V}_{\delta \gamma} = -\frac{2k_1k_2 \Delta^2}{\sqrt{1+4k_2^2 \Delta^2}} \frac{\partial \Delta}{\partial \delta} - 2k_4 q^2 z^2, \label{eq:V_dot_delta_gamma_backstepping}
\end{equation}
which considering \eqref{eq:Delta_globa} and $\partial \Delta / \partial \delta = 1$, \eqref{eq:V_dot_delta_gamma_backstepping} is negative definite on $\mathcal{T}$ and considering \eqref{eq:Delta_Barfli} and $\partial \Delta / \partial \delta = 1 + \tan^2 \frac{\delta}{2} >0$  for all $|\delta| < \pi $, \eqref{eq:V_dot_delta_gamma_backstepping} is negative definite on $\mathcal{T}_1$.
Based on \cite[Lemma 4.3]{khalil_nonlinear_2002},  
there exists 
$\alpha_3\in\mathcal{K}$, such that, $\dot{V}_{\delta \gamma} \le - \alpha_3 \circ \alpha_2^{-1}
(V_{\delta \gamma})$, where $
\alpha_3 \circ \alpha_2^{-1}\in\mathcal{K}$. 
The Lyapunov function $V_{\delta \gamma}$ satisfies the assumptions of Proposition \ref{prop:composite_Lyap_function}, which implies that the composite Lyapunov functions $V(\rho, \delta, \gamma) = \mathcal{V}(\rho^2, V_{\delta \gamma})$ and $V(\rho, \delta, \gamma) = \mathcal{V}(V_{\delta \gamma},\rho^2)$ with  $\dot{V}(\rho,\delta,\gamma) \le - \alpha(V)$ are strict CLFs for \eqref{eq:unicycle_polar_closed_loop-Gv-1}. 
Furthermore, $\dot{V}(\rho,\delta,\gamma) \le - \alpha(V)$ implies the existence of  $\beta\in\mathcal{KL}$  such that
$V(t) \le \beta(V_0, t)$  
holds for all $t\geq 0$ according to \cite[Lemma 4.4]{khalil_nonlinear_2002}. Considering the fact from Proposition \ref{prop:composite_Lyap_function} that the composite Lyapunov functions are bounded as $ \bar\alpha_1(|(\rho,\delta, \gamma)|_{\hat{\mathcal{S}}}) \le V(\rho,\delta,\gamma) \le \bar\alpha_2(|(\rho,\delta, \gamma)|_{\hat{\mathcal{S}}})$, from  $V(t) \le \beta(V_0, t)$, we have that     
$\abs{(\rho,\delta,\gamma)}_{\hat{\mathcal{S}}} \le \bar{\alpha}_1^{-1}(\beta( \bar{\alpha}_2(\abs{(\rho_0, \delta_0,\gamma_0)}_{\hat{\mathcal{S}}}),t))$
where $\bar{\alpha}_1^{-1}(\beta( \bar{\alpha}_2(r),t))$ is class $\mathcal{KL}$ in $(r,t)$.
Thus, 
$\rho =  \delta = \gamma = 0$ is GAS on $\hat{\mathcal{S}}$ where $\hat{\mathcal{S}} = \{\mathcal{S}, \mathcal{S}_1\}$ is dependent on the choice of $\Delta$.
\end{proof}

The globally asymptotically stabilizing backstepping controller \eqref{eq:backstepping_controllers} with \eqref{eq:Delta_globa} is referred to as GloBa (pronounced `globa') 
employs the backstepping transformation \eqref{eq:backstepping_z} and the associated strict CLF \eqref{eq:V_CLF_backstepping} defined with \eqref{eq:Delta_globa} and its region of attraction is the unconstrained state-space  $\mathcal{T}$ given in \eqref{eq:ss_S}.
 We also present the backstepping controller \eqref{eq:backstepping_controllers}
which employs the backstepping transformation \eqref{eq:backstepping_z} and the associated CLF \eqref{eq:V_CLF_backstepping} defined with \eqref{eq:Delta_Barfli}.
The region of attraction of this controller is the constrained state-space $\mathcal{S}_1$ given in \eqref{eq:ss_S1}, which is combined with \eqref{eq-basic-v-control} and \eqref{eq-omega-general}
entails all positions except the nonnegative half of the $x$-axis. In other words, the algorithm achieves stable parking while never crossing directly in front of the parking target. Since this controller maintains $|\delta(t)|<\pi$ for the polar angle $\delta$, implying the unicycle never crosses the line in front of the target, we refer to the controller as the Backstepping to Avoid Running across Front Line (BAR-FLi, pronounced `Bar Fly') controller.
It is important to note that the CLF 
$V = \rho^2 + V_{\delta\gamma}$ where $V_{\delta \gamma}$ is as in \eqref{eq:V_CLF_backstepping} with \eqref{eq:Delta_globa} being the same CLF that Restrepo et al.~\cite{restrepo2020leader} have presented, though the controller differs substantially from \eqref{eq-basic-v-control} and \eqref{eq:backstepping_controllers} with \eqref{eq:Delta_globa}.

\subsection{Passivity-inspired design}

Let us regard the system \eqref{eq:unicycle_polar_closed_loop-Gv-3-pass} as a connection of the integrator subsystem \eqref{eq:unicycle_polar_closed_loop-Gv-3n-pass}, which is passive (positive real) and driven by the output $y_2={k_1} \sin(2\gamma)$ of the subsystem \eqref{eq:unicycle_polar_closed_loop-Gv-3a-pass}. Based on the theorem on the stability of a negative feedback interconnection of stable systems, the design objective is that of choosing the control $\tilde\omega$ to make the subsystem \eqref{eq:unicycle_polar_closed_loop-Gv-3a-pass} strictly passive, with $-y_1$, where $y_1 = k_3\delta$ and $k_3>0$ is the negated output of the subsystem \eqref{eq:unicycle_polar_closed_loop-Gv-3n-pass} serving as the input to \eqref{eq:unicycle_polar_closed_loop-Gv-3a-pass}, and to make the overall feedback loop zero-state observable with respect to $\gamma$ as the output of the overall feedback loop. The next two designs pursue this approach on the state spaces $\mathcal{T}_2$ and $\mathcal{T}_3$ defined in \eqref{eq:ss_S2} and \eqref{eq:ss_S3}, respectively. Compared to the backstepping approach in the previous section, both designs impose bounds on the LoS angle $\gamma$, and one of them additionally constrains the polar angle $\delta$—a feature not achievable with backstepping. To derive our results, we first present the following lemma.

\begin{lemma}\label{lemma:sinterm_upperbound}
The inequality $1- k \cos\gamma(1+\cos\gamma) \leq 2 (1 + k)  \tan^2\frac{\gamma}{2}$ holds for all $k\geq 1$ and $x\in\mathbb{R}$.
\end{lemma}

\paragraph{BoLSA controller ($\mathcal{T}_2$)} 

\begin{theorem}
\label{thm:unicycle_CLF_BoLSA}
Consider the system \eqref{eq:unicycle_polar_closed_loop-Gv-1} in closed-loop with \eqref{eq-basic-v-control}, \eqref{eq-omega-general} and 
\begin{equation}
\label{eq-control-bounded-in-gamma}
\tilde\omega =  k_2\sin\gamma + k_3 \displaystyle
\frac{\cos\gamma}{\left(1+\displaystyle\tan^2\frac{\gamma}{2}\right)^2}\delta \,, 
\end{equation}
with  $k_1, k_2, k_3 > 0$ such that 
$k_1 k_3\geq k_2^2$.
The origin $\rho = \delta = \gamma = 0$ is GAS on $\mathcal{S}_2$ in accordance with Definition \ref{def-our-GAS}.
Furthermore, all the composite Lyapunov functions $V(\rho, \delta, \gamma) = \mathcal{V}(\rho^2, V_{\delta \gamma})$ and $V(\rho, \delta, \gamma) = \mathcal{V}(V_{\delta \gamma},\rho^2)$, for all functions $\mathcal{V}$ satisfying the conditions in Proposition \ref{prop:composite_Lyap_function}, and with $V_{\delta \gamma}$ defined as
\begin{equation}
\hspace*{-0.2cm}
\label{eq:CLF_Bolsa_delta_gamma}
V_{\delta \gamma}(\delta, \gamma) =  k_3  \left(1+
\frac{2q^2+U}{2qk_2}\right)U
+   \left (\delta+2q\tan \frac{\gamma}{2} \right)^2\,,
\end{equation}
where
\begin{align}
U=\delta^2+  4q^2 \tan^2\frac{\gamma}{2} \,, \quad q = \sqrt{\frac{k_1}{k_3}}
\end{align} 
are (globally) strict CLFs for \eqref{eq:unicycle_polar_closed_loop-Gv-1} with respect to the input pair $(v/\rho,\omega)$ in the sense of Definition \ref{def-CLF}. 
\end{theorem}

\begin{proof}
We start by developing the feedback law \eqref{eq-control-bounded-in-gamma} and the CLF \eqref{eq:CLF_Bolsa_delta_gamma} for the $(\delta, \gamma)$-subsystem   \eqref{eq:unicycle_polar_closed_loop-Gv-3-pass}. 
We are going to use the following Lyapunov terms:
\begin{eqnarray}
U &=& \delta^2 +   q^2 V_0, \quad V_0 = 4 \tan^2\frac{\gamma}{2}\,,  \label{eq:U_bolsa} \\
\Pi_0 &=& z^2\,, \quad z = \delta + 2 q \tan\frac{\gamma}{2}.
\end{eqnarray}
The time derivative of the Lyapunov expression \eqref{eq:U_bolsa} along the solutions of \eqref{eq:unicycle_polar_closed_loop-Gv-3-pass} is
\begin{equation}
\dot U  =  \frac{8 q^2 \sin\gamma}{(1+\cos\gamma)^2} \left[ k_3 \delta \frac{(1+\cos\gamma)^2}{4} \cos\gamma 
- \tilde{\omega}\right]\,.
\end{equation}
We pick the control as in \eqref{eq-control-bounded-in-gamma}, where we use the fact that $(1+\cos\gamma)^2/4 = 1/(1+\tan^2 \gamma/2)^2$. With this feedback, we get
\begin{equation}
\dot U  = -2 \frac{k_2  q^2 4 \sin^2\gamma}{(1+\cos\gamma)^2} = - 2 k_2 q^2 V_0\,,
\end{equation}
and the dynamics for the LoS angle as
\begin{equation}
\dot\gamma = - k_2 \sin\gamma -  k_3 \delta\frac{(1+\cos\gamma)^2}{4} \cos\gamma  \,.
\end{equation}
Next, we consider the dynamics of the error $z$, which with
 $\sin(\gamma)\cos(\gamma) = \cos\gamma(1+\cos\gamma) \tan\frac{\gamma}{2}$ and $\frac{\sin\gamma}{1+\cos\gamma} =\tan\frac{\gamma}{2} $, are given by
\begin{equation}
\hspace*{-0.3cm}
\dot z = -k_2z + k_2\delta + k_3 q \frac{\cos\gamma(1+\cos\gamma)}{2}\left( 2 q \tan\frac{\gamma}{2}-\delta\right)\,.    
\end{equation}
Then,
\begin{equation}
\hspace*{-0.2cm}
\begin{aligned}[b]
\dot \Pi_0 = - 2k_2z^2 + 2k_2 z \delta &+ 
4 k_1 q \cos\gamma(1+\cos\gamma) \tan^2\frac{\gamma}{2} \\
&  
-k_3 q \frac{\cos\gamma(1+\cos\gamma)} {2}\delta^2\,.
\end{aligned}
\end{equation}
With $z\delta \leq  \frac{z^2}{4} + \delta^2$ and $\cos\gamma(1+\cos\gamma) \tan^2\frac{\gamma}{2} \leq  \frac{V_0}{2}
$, we get
\begin{equation}
\begin{aligned}[b]
    \dot \Pi_0 \le - &\frac{3}{2}k_2 \Pi_0 + 2 k_ 1 q V_0 \\
&+ 2 k_2 \delta^2\left(1- q \frac{k_3 }{2 k_2} 
\cos\gamma(1+\cos\gamma)\right)\,.
\end{aligned}
\end{equation}
It then follows from Lemma~\ref{lemma:sinterm_upperbound} with $k_1 k_3 \ge  k_2^2$ that 
\begin{equation}
\dot\Pi_0 
\leq - \frac{3}{2} k_2 \Pi_0 + 2 k_1 q  V_0 +   2 k_3 q \delta^2 V_0.
\end{equation}
Taking into account that $\frac{k_3 }{k_2} q \dot U = - 2k_2 q^2 V_0$, and $\frac{k_3}{2q k_2 } \dot{U}^2 = -2k_3 q \delta^2 V_0 - 2k_1 q V_0^2$ and denoting
\begin{equation}
\Pi_1 = \Pi_0 + \frac{k_3 }{k_2} q  U  + \frac{k_3 }{2 q k_2} U^2\,, \label{eq:Pi_1_bolsa}
\end{equation}
we get
\begin{equation}
\dot\Pi_1 \leq - \frac{3}{2}k_2\Pi_0 - 2 k_1 q V_0^2\,.
\end{equation}
In conclusion, from \eqref{eq:U_bolsa} and \eqref{eq:Pi_1_bolsa}, we construct the Lyapunov function for the $(\delta, \gamma)$-subsystem as $V_{\delta \gamma} = k_3 U + \Pi_1$ which is equivalent to \eqref{eq:CLF_Bolsa_delta_gamma} 
and has the time derivative \begin{equation}
\dot V_{\delta \gamma} \le  - 2k_1 k_2 V_0 - \frac{3}{2}k_2 \left(\delta + q \tan \frac{\gamma}{2}\right)^2 - 2k_1 q V_0^2
\label{eq:V_dot_Bolsa}
\end{equation} 
which is negative definite on $\mathcal{T}_2$. Analogous to the proof of Theorem~\ref{thm:backstepping_theorem}, we arrive at both statements of Theorem~\ref{thm:unicycle_CLF_BoLSA}.
\end{proof}

The controller \eqref{eq-control-bounded-in-gamma}
is bounded in the LoS angle $\gamma$ and we, consequently, refer to it as the Bounded-in-LoS Angle (BoLSA, pronounced `bolsa') controller. Combined with \eqref{eq-basic-v-control} and \eqref{eq-control-bounded-in-gamma}, BoLSA's region of attraction is the state space 
$\mathcal{S}_2$ defined in \eqref{eq:ss_S2},
which means all the initial headings except exactly away from the target. In other words, the algorithm achieves stable parking while, in the process, never ``turning its back'' against the target position.


\paragraph{BAgAl controller ($\mathcal{T}_3$)}

\begin{theorem}
\label{thm:unicycle_CLF_Bagal}
Consider the system \eqref{eq:unicycle_polar_closed_loop-Gv-1} in closed-loop with \eqref{eq-basic-v-control}, \eqref{eq-omega-general} and 
\begin{equation}
\label{eq-control-bounded-in-gamma-delta}
\tilde\omega = k_2\sin\gamma + \displaystyle
\frac{ 2k_3 \cos\gamma }{\left(1+\displaystyle\tan^2\frac{\gamma}{2}\right)^2} \left(1+\tan^2 \frac{\delta}{2} \right)\tan \frac{\delta}{2}\,,
\end{equation}
with  $k_1, k_2, k_3 > 0$ such that 
$k_1 k_3\geq k_2^2$.
The origin $\rho = \delta = \gamma = 0$ is GAS on $\mathcal{S}_3$ in accordance with Definition \ref{def-our-GAS}.
Furthermore, all the composite Lyapunov functions $V(\rho, \delta, \gamma) = \mathcal{V}(\rho^2, V_{\delta \gamma})$ and $V(\rho, \delta, \gamma) = \mathcal{V}(V_{\delta \gamma},\rho^2)$, for all functions $\mathcal{V}$ satisfying the conditions in Proposition \ref{prop:composite_Lyap_function}, and with $V_{\delta \gamma}$ defined as
\begin{equation}
\hspace*{-0.2cm}
\label{eq:V_CLF_Bagal}
   V_{\delta \gamma}(\delta, \gamma) =    a\left[ (1+U)^3 -1\right]+ 
\left (\tan\frac{\delta}{2} + q\tan\frac{\gamma}{2}\right)^2,
\end{equation}
where
\begin{align}
U=\tan^2 \frac{\delta}{2}+ q^2 \tan^2 \frac{\gamma}{2}, \quad
q = \sqrt{ \frac{ k_1}{k_3}}, 
\end{align} 
with $a = \max\{k_1 q, 2 \sqrt{k_1k_2}\}/ 3k_2 q^2$, are (globally) strict CLFs for \eqref{eq:unicycle_polar_closed_loop-Gv-1} with respect to the input pair $(v/\rho,\omega)$ in the sense of Definition \ref{def-CLF}. 
\end{theorem}

The controller \eqref{eq-control-bounded-in-gamma-delta}
is bounded with respect to polar angle $\delta$ and the LoS angle $\gamma$. We refer to it as the Bounding-Angles Algorithm Design (BAgAl, pronounced “bagel”) controller. Combined with \eqref{eq-basic-v-control} BAgAl's, region of attraction is the state space $\mathcal{S}_3$ defined in \eqref{eq:ss_S3}
which includes all initial headings except those facing exactly opposite of the target, and all positions except those on the nonnegative half of the $x$-axis. In other words, the algorithm achieves stable parking while never “turning its back” on the target position or crossing directly in front of it.

\section{Barrier CLFs and ``nearly global'' feedbacks}\label{sec:barrier}

\begin{figure}[t]
\centering
\includegraphics[width=0.5\linewidth]{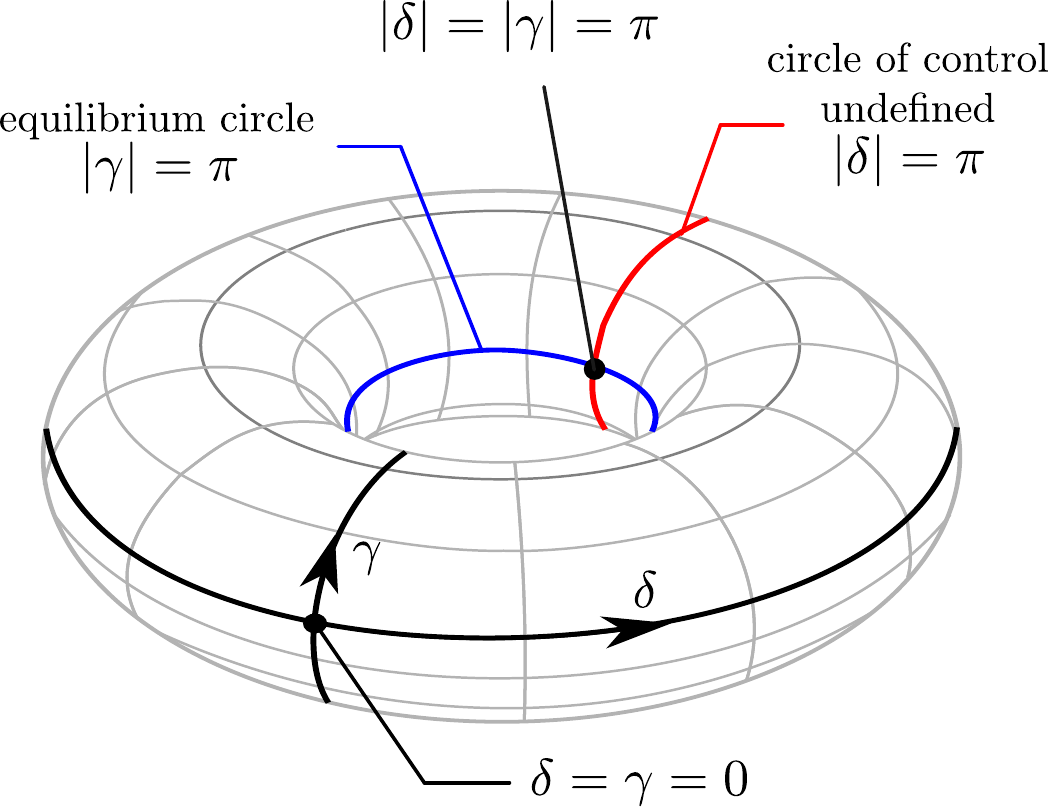}
\caption{Visualization of the torus $\{|\delta|\leq\pi\}\times \{|\gamma|\leq \pi\}$ with the undefined control set $\{|\delta|=\pi\}\times \{|\gamma|\leq \pi\}$ (red) and the equilibrium set $\{|\delta|<\pi\}\times \{|\gamma|= \pi\}$ (blue).}
\label{fig:torus}
\end{figure}

\paragraph{Barrier CLFs} The Lyapunov function \eqref{eq:CLF_Bolsa_delta_gamma}
blows up at $\gamma=\pm \pi$ and the Lyapunov function 
\eqref{eq:V_CLF_backstepping} with \eqref{eq:Delta_Barfli} blows up at $\delta=\pm \pi$.  Similarly, the Lyapunov function \eqref{eq:V_CLF_Bagal} blows up on the boundary of $\mathcal{T}_3$, the square $\{\delta=\pm\pi\}\cup \{\gamma=\pm\pi\}$. 
Such Lyapunov functions are called ``barrier Lyapunov functions'' and they ensure the invariance of the sub-level sets where $|\gamma|$ and $|\delta|$, respectively, are bounded. 

\paragraph{Feedback actions near and at the barriers} 

The three controllers employing barrier CLFs exclude the measure zero sets $\gamma=\pm \pi$, $\delta=\pm \pi$, and $\{\delta=\pm\pi\}\cup \{\gamma=\pm\pi\}$, from their respective regions of attraction. In \eqref{eq-control-bounded-in-gamma}, the bounded feedback term $\sin\gamma$ in relation to the {\em barrier} Lyapunov term $\tan^2\frac{\gamma}{2}$ creates unstable equilibria on the set $\gamma=\pm\pi$, while in \eqref{eq:backstepping_controllers} with \eqref{eq:Delta_Barfli}, the superlinear term $\left(1+\tan^2\frac{\delta}{2}\right)\tan\frac{\delta}{2}$ grows to infinity to prevent $\delta$ from reaching its boundary, precluding the existence of equilibria on the set $\delta = \pm \pi$. The combined law \eqref{eq-control-bounded-in-gamma-delta} incorporates both terms to keep both states within the interval $(-\pi,\pi)$. The substantially different behavior among the feedback terms, all of which relate to barrier Lyapunov functions of the form $\tan^2$, is because the control input $\omega$ affects $\gamma$ directly but influences $\delta$ only through an integrator (a difference of one relative degree with respect to $\omega$).
The feedback law \eqref{eq-control-bounded-in-gamma}, at $\gamma = \pm \pi$ and $\delta = 0$ (when the vehicle points away from both the target position and the target heading), drives the system toward the target position $\rho = 0$ through the reverse velocity feedback $v = -k_1\rho$, but does not achieve the target heading. In contrast, the feedback laws 
\eqref{eq:backstepping_controllers} with \eqref{eq:Delta_Barfli} and \eqref{eq-control-bounded-in-gamma-delta}, starting at $\delta=\pm \pi$, the steering input is $\omega = \pm \infty$, i.e., undefined. This is due to the topological impossibility of global stabilization on 
the cylinder $ S^1\times \mathbb{R}$ for \eqref{eq:backstepping_controllers} with \eqref{eq:Delta_Barfli} and on 
the torus $S^1 \times S^1$ for \eqref{eq-control-bounded-in-gamma-delta}. 


\paragraph{Topological behavior of BAgAl controller \eqref{eq-control-bounded-in-gamma-delta}} The value of the BAgAl controller \eqref{eq-control-bounded-in-gamma-delta} is undefined on the set $\{\rho>0\}\times \{|\delta|=\pi\}\times \{|\gamma|\leq \pi\} \ = \ \{x>0\}\cap \{y=0\}$, the red circle in Fig.~\ref{fig:torus} which is a part of the boundary of the open state space $\mathcal{T}_3$ and a measure zero subset of the ``full configuration space'' $\mathcal{W} := \{\rho>0\}\times \{|\delta|\leq\pi\}\times \{|\gamma|\leq \pi\} \ = \ \{x^2+y^2 > 0\}$. The set $\{|\delta|<\pi\}\times \{|\gamma|= \pi\}$, the blue circle in Fig~\ref{fig:torus} which is also a part of the boundary of the open state space $\mathcal{T}_3$ and a measure zero subset of the torus $\{|\delta|\leq\pi\}\times \{|\gamma|\leq \pi\}$, is an equilibrium set of \eqref{eq:unicycle_polar_closed_loop-Gv-3-pass}, \eqref{eq-control-bounded-in-gamma-delta}, on which the vehicle faces away from the positional target but may be anywhere except $x>0,y=0$. The set $\{\rho>0\}\times \{|\delta|<\pi\}\times \{|\gamma|= \pi\} \ = \ \{x<0\}\cup \{y\neq 0\}$, which is of measure zero in $\mathcal{W}$, consists of straight-line trajectories of \eqref{eq:unicycle_polar_closed_loop-Gv-1}, \eqref{eq-basic-v-control}, \eqref{eq-omega-general}, \eqref{eq-control-bounded-in-gamma-delta} which ``back up'' to the target position with $\theta\neq 0$. From the entire configuration space $\mathcal{W}$, 
only the measure zero subset $\{\rho>0\}\times (\{|\delta|=\pi\}\cup \{|\gamma|= \pi\})$, which are the intersecting blue and red circles on the torus in Fig.~\ref{fig:torus}, is excluded from the region of attraction of $\rho=\delta=\gamma=0$. 

\begin{figure}[t]
\centering
\includegraphics[width=0.9\linewidth]{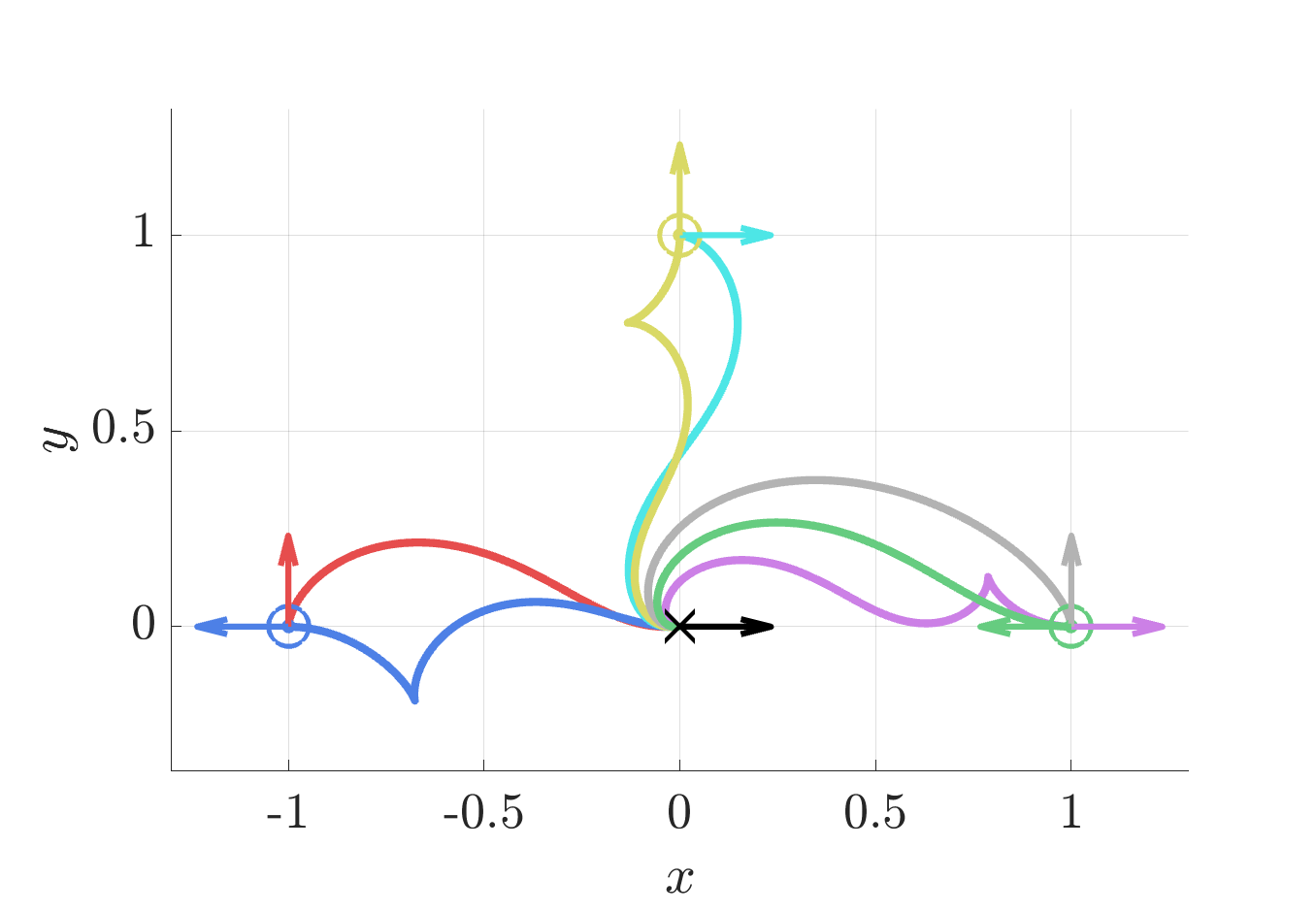}
\caption{Closed-loop trajectories of GloBa~\eqref{eq:backstepping_controllers} with \eqref{eq:Delta_globa}.}
\label{fig:genova_vs_globa}
\end{figure}


\begin{figure}[t]
\centering
\includegraphics[width=0.8\linewidth]{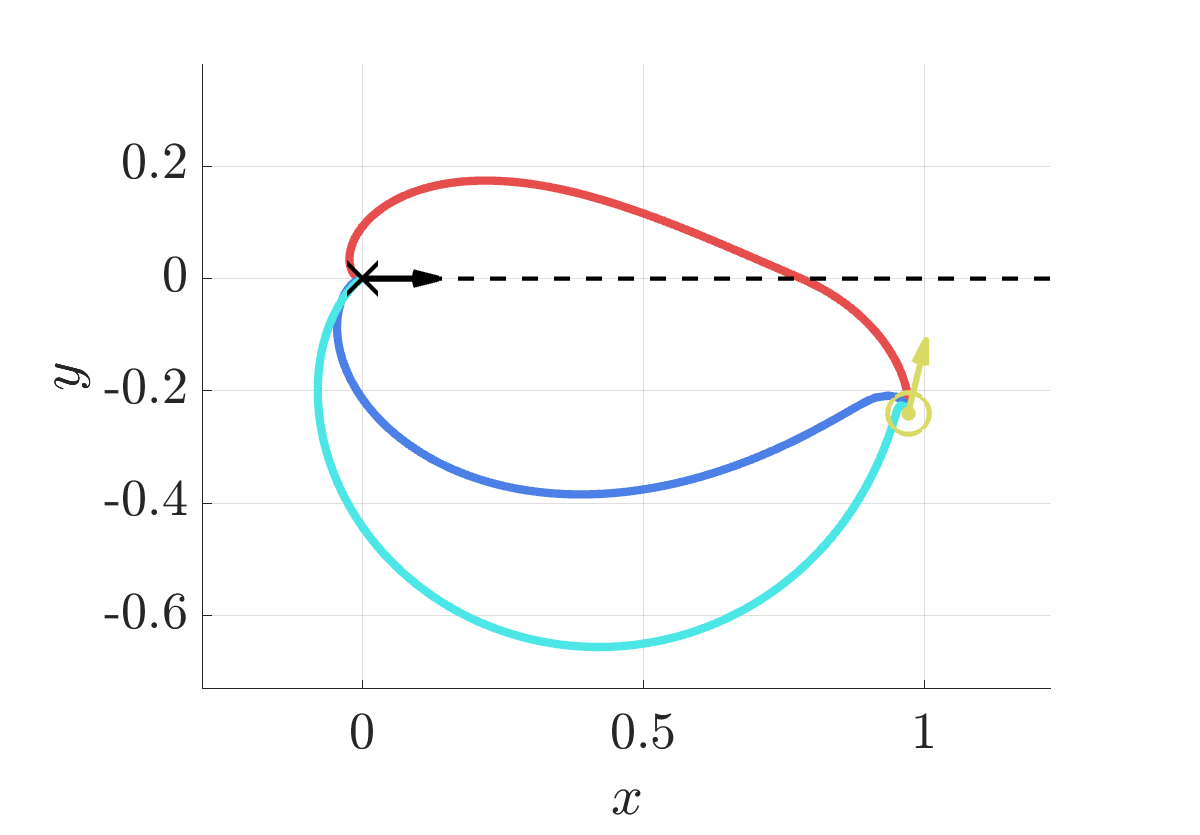}
\caption{Closed-loop trajectories with BAR-FLi (blue) controller~\eqref{eq:backstepping_controllers} with \eqref{eq:Delta_Barfli} and BAgAl (cyan) controller~\eqref{eq-control-bounded-in-gamma-delta} compared to the GloBa controller~~\eqref{eq:backstepping_controllers} with \eqref{eq:Delta_globa} (red).}
\label{fig:barfli_abound}
\end{figure}

\section{Comparing trajectories}

Throughout the simulation results, the target position and heading angle are shown in black.
 We show the unconstrained GloBa~\eqref{eq:backstepping_controllers} with \eqref{eq:Delta_globa} with all gains set to one in Fig.~\ref{fig:genova_vs_globa}. 
 Next, Fig.~\ref{fig:barfli_abound} compares BAR-FLi~\eqref{eq:backstepping_controllers} with \eqref{eq:Delta_Barfli} (blue) and BAgAl (cyan) against GloBa (red), with gains $(k_1,k_2,k_3,k_4) = (1,1,0.1,1)$. Here, both BAR-FLi and BAgAl avoid crossing in front of the target, consistent with the intended effect of the barrier CLFs on $\delta$. Since the trajectories by BAgAl and BoLSA \eqref{eq-control-bounded-in-gamma} are essentially identical, BoLSA is omitted from the plots.

\bibliographystyle{IEEEtranS}
\bibliography{root}

\clearpage

\end{document}